\documentclass[11pt, twoside, reqno]{amsart}
\usepackage{verbatim,amscd,amsmath,amsthm,amstext,amssymb,hyperref,extpfeil,dsfont, enumerate, tikz, tikz-cd, mathtools
}
\usepackage{graphicx}
\usepackage[outdir=./]{epstopdf}
\usepackage{floatflt}
\usepackage{float} 
\theoremstyle{plain}

 \newtheorem{thm}{Theorem}[section]
 \newtheorem{cor}[thm]{Corollary}
 \newtheorem{lem}[thm]{Lemma}
 \theoremstyle{remark}
 \newtheorem{rem}[thm]{Remark}

 \newcommand{\Jac}{\mathrm{Jac}}
 \newcommand{\ord}{\mathrm{ord}}
 \newcommand{\mult}{\mathrm{mult}}
 \newcommand{\C}{\mathbb{C}}
 \newcommand{\R}{\mathbb{R}}
 \newcommand{\Z}{\mathbb{Z}}

\begin{document}

\title[On a property of Fermi curves]{On a property of Fermi curves of $2$-dimensional periodic Schr\"odinger operators}

\author[E.~L\"ubcke]{Eva L\"ubcke}

\address{Mathematics Chair III\\
	Universit\"at Mannheim\\
	D-68131 Mannheim\\ 
	Germany}

\email{eluebcke@mail.uni-mannheim.de}


\date{June 8, 2016}

\begin{abstract}
	We consider a compact Riemann surface with a holomorphic involution, two marked fixed points of the involution and a divisor obeying an equation up to linear equivalence of divisors involving all this data. Examples of such data are Fermi curves of $2$-dimensional periodic Schr\"odinger operators. We show that the equation has a solution if and only if the two marked points are the only fixed points of the involution.
\end{abstract}
\maketitle
\section{Introduction}
Let $X$ be a  compact Riemann surface of genus $g<\infty$ and $\sigma:X\to X$ a holomorphic involution with two fixed points $P_1$ and $P_2$. Then the linear equivalence $D+\sigma(D)\simeq K + P_1+P_2$ is solvable by a divisor $D$ of degree $g$ if and only if $P_1$ and $P_2$ are the only fixed points of $\sigma$. 
It is known that on Fermi curves of $2$-dimensional periodic Schr\"{o}dinger operators, there exists 
a holomorphic involution $\sigma$ with two fixed points such that the pole divisor of the corresponding normalized eigenfunctions obeys this linear equivalence, see \cite{N-V}. It was remarked in \cite{N-V} without proof that I.R. Shafarevich and V. V. Shokurov pointed out that  $D+\sigma(D)\simeq K + P_1+P_2$ can hold if and only if $P_1$ and $P_2$ are the only fixed points of $\sigma$. To prove this assertion, we basically use the results reflecting the connection between the Jabobian variety and the Prym variety which was shown in \cite{Mu}. Since we are mainly using the ideas shown there and not the whole concept, we will explain later how this connection involves here. 
\section{A two-sheeted covering}
\noindent Let $X$ be a compact Riemann surface, $\sigma$ a holomorphic involution on $X$ and $P_1,P_2\in X$ fixed points of $\sigma$, i.e.  $\sigma(P_i)=P_i$ for $i=1,2$. For $p,q\in X$ let $p \sim q :\Leftrightarrow (p=q \lor p=\sigma(q))$
and define $X_\sigma:= X/\sim$. Let $\pi: X\rightarrow X_\sigma$ be the canonical two-sheeted covering map. Since the subgroup of $\mathrm{Aut}(X)$ generated by $\sigma$ is $\Z_2$, $X_\sigma$ is a compact Riemann surface and $\pi: X\to X_\sigma$ is holomorphic, compare \cite[Theorem III.3.4.]{Mir}. Due to the construction of $X_\sigma$, the fixed points of $\sigma$ coincide with the ramification points of $\pi$. 
The set of ramification points of $\pi$ we denote by $r_\pi$. Then the map $\pi$ is locally biholomorphic on $X\setminus r_\pi$, see \cite[Corollary I.2.5]{Fo}. We define the ramification divisor of $\pi$ on $X$ as $R_{\pi}:=\sum_{p\in r_\pi}p$. In general, the ramification divisor is defined as $R_{\pi}:=\sum_{p\in X}(\mathrm{mult}_{p}(\pi)-1)\cdot  p$, where  the multiplicitiy $\mathrm{mult}_{p}(\pi)$ of $\pi$ in $p$ denotes the number of sheets which meet in $p$, compare \cite[Definition II.4.2]{Mir}. Since $\mathrm{mult}_{ p}(\pi)=1$ for $p\in X\setminus r_\pi$ and $\mathrm{mult}_{ p}(\pi)=2$ for $p\in r_\pi$, this coincides with the above definition.
Furthermore, $b_\pi:=\pi[r_\pi]$ is the set of branchpoints of $\pi$ on $X_\sigma$.
The involution $\sigma$ extends to an involution on the divisors on $X$ by $\sigma\big(\sum_{p\in X} a(p) p\big):= \sum_{p\in X}a(p) \sigma(p)$ which we also denote as $\sigma$. So the degree of a divisor is conserved under $\sigma$. We define the pullback of a point $p_\sigma\in X_\sigma$ as \[\pi^\ast p_\sigma := \sum_{p\in \{\pi^{-1}[\{p_\sigma\}]\}} \mult_{p} (\pi) p.\] With this definition, the pullback of a divisor $D:=\sum_{p_\sigma\in X_\sigma}a(p_\sigma) p_\sigma$ on $X_\sigma$ is defined as $\pi^\ast D:= \sum_{p_\sigma\in X_\sigma} a(p_\sigma)\pi^\ast p_\sigma$.
Since $\pi$ is a non-constant holomorphic map between two Riemann surfaces, every meromorphic $1$-form on $X_\sigma$ can be pulled back to a meromorphic $1$-form $\omega:=\pi^\ast\omega_\sigma$ on $X$, compare for example \cite[Section IV.2.]{Mir}. 
\begin{lem}\label{lem:divisor}
	Let $X, X_\sigma$ and $\pi$ be given as above and let $\omega_\sigma$ be a non-constant meromorphic $1$-form on $X_\sigma$.
	\begin{enumerate}[(a)]
	\item The divisor of $\pi^\ast \omega_\sigma$ on $X$ is given by 
	$
	(\pi^\ast \omega_\sigma)= \pi^\ast(\omega_\sigma) + R_\pi.
	$
	\item Let $g_\sigma$ be the genus of $X_\sigma$. Then there exists a divisor $\widetilde K$ on $X$ with $\deg (\widetilde K)=2g_\sigma-2$ such that $(\pi^\ast\omega_\sigma)= \widetilde K+\sigma(\widetilde K)+R_\pi$.
	\end{enumerate}
\end{lem}
\begin{proof}
	\begin{enumerate}[(a)]
		\item Due to \cite[Lemma IV.2.6]{Mir} one has for 
		$p\in X$ that 
		\[
				\mathrm{ord}_{p}(\pi^\ast \omega_\sigma)= (1+\mathrm{ord}_{\pi(p)}(\omega_\sigma))\mathrm{mult}_{p}(\pi)-1
				\] 
		with $\mathrm{ord}_p(\pi^\ast\omega_\sigma)$ as defined in \cite[Section IV.1.9]{Mir}.
				Inserting this into the definition of $(\pi^\ast \omega_\sigma)= \sum_{p \in X} (\ord_p(\pi^\ast \omega)) $  
				yields the assertion.
		\item One has $\deg K_\sigma=\deg(\omega_\sigma)=2g_\sigma-2$ where $K_\sigma$ is the canonical divisor on $X_\sigma$. Let $p_\sigma\in X_\sigma$ be a point in the support of $(\omega_\sigma)$ as defined in \cite[Section V.1]{Mir}. For $p_\sigma\not\in b_\pi$ one has $\pi^\ast p_\sigma= p + \sigma(p)$ with $p\neq \sigma(p)\in X$ and for $p_\sigma\in b_\pi$ it is $\pi^\ast p_\sigma=2p$ with $p\in r_\pi$. 
		For $p_\sigma\not\in b_\pi$, let one of the pulled back points in  $\pi^\ast p_\sigma$ be the contribution to $\widetilde K$ and for $p_\sigma\in b_\pi$ the pulled back point is counted with multiplicity one in $\widetilde K$. Then $\pi^\ast(K)=\widetilde K+\sigma(\widetilde K)$ and the claim follows from (a). 
	\end{enumerate}
\end{proof}
	Now we are going to construct a symplectic cycle basis of $H_1(X,\Z)$ from a symplectic cycle basis of $H_1(X_\sigma,\Z)$. The holomorphic map $\sigma:X\to X$ induces a homomorpism of $H_1(X,\Z)$ which we denote as
	\[
	\sigma_\sharp: H_1(X,\Z)\to H_1(X,\Z), \quad  \gamma\mapsto \sigma_\sharp\gamma.
	\] 
	Let $g_\sigma$ be the genus of $X_\sigma$ and $A_{\sigma,1},\ldots, A_{\sigma,g_\sigma}$, $B_{\sigma,1}, \ldots, B_{\sigma,g_\sigma}$ be represantatives of a symplectic basis of $H_1(X_\sigma,\Z)$, i.e. 
	\[A_{\sigma,i}\star A_{\sigma,\ell}=B_{\sigma,i}\star B_{\sigma,\ell}=0 \ \mbox{ and } \  A_{\sigma,i}\star B_{\sigma,\ell}=\delta_{i\ell},\]
	where $\star$ is the intersection product between two cycles.
    From Riemann surface theory it is known that such a basis exists, compare e.g. \cite[Section VIII.4]{Mir}.
	Due to Hurwitz's Formula, e.g. \cite[Theorem II.4.16]{Mir}, one knows that that $\sharp b_\pi=2n$ is even for a two sheeted covering $\pi:X\to X_\sigma$ and that the genus $g$ of $X$ is given by $g=2g_\sigma+n-1$. Hence a basis of $H_1(X,\Z)$ consists of $4g_\sigma+2n-2$ cycles. The aim is to construct a symplectic basis of $H_1(X,\Z)$ which we denote as $A_i,\sigma_\sharp A_i, B_i,\sigma_\sharp B_i$ and $C_j,D_j$ and which has the following two properties: 
    First of all, the only non-trivial pairwise intersections between elements of the basis of  $H_1(X,\Z)$ must be given by
	   \begin{equation}\label{eqn:canonical_zero_prym}
	   A_i \star  B_i = \sigma_\sharp  A_i \star \sigma_\sharp  B_i=  C_j \star  D_j = 1.
	   \end{equation}
	Secondly, the involution $\sigma_\sharp$ has to map $A_i$ to $\sigma_\sharp A_i$ and vice versa, $B_i$ to $\sigma_\sharp B_i$ and vice versa and has to act on $C_j$ and $D_j$ as $\sigma_\sharp C_j=-C_j$ and $\sigma_\sharp D_j=-D_j$.
	Here, and from now on, we consider $i,\ell\in \{1,\dots, g_\sigma\}$ and $j,k\in \{1,\dots,n-1\}$ as long as not pointed out differently. 
	The difference in the notation of the cycles indicates the origin of these basis elements: the $A$- and $B$-cycles on $X$ will be constructed via lifting a certain symplectic cycle basis of $H^1(X_\sigma,\Z)$ via $\pi$ and the $C$- and $D$-cycles originate from the branchpoints of $\pi$.\vspace{.1cm}\\ 
	\begin{minipage}{.55\textwidth}
		We will start by constructing the $C$- and $D$-cycles. A sketch of the idea how to do this is shown for $n=3$ and $g_\sigma=0$ in figure \ref{fig:zykel} . We connect the points in $b_\pi$ pairwise by paths $s_j$ for $j=1,\dots,n$. The set of points corresponding to a path $s_j:[0,1]\to X_\sigma$ we denote by $[s_j]:=\{s_j(t)\,\mid \, t\in [0,1]\}$ and use the same notation for any other path considered as a set of points in $X$ or $X_\sigma$. Let $[s_j]^\circ$ be the corresponding set with $t\in (0,1)$. The paths $s_j$ are constructed in such a way that every branchpoint is connected with exactly one other branchpoint and such that $s_k\cap s_j=\emptyset$ for $k\neq j$. This is possible since the branchpoints lie discrete on $X_\sigma$: suppose the first two branchpoints are connected by $s_1$ such that $s_1$ contains no other branchpoint. Then one can find a small open tubular neighborhood $N(s_1)$ of $s_1$ in $X_\sigma$ with boundary $\partial N(s_1)$ in $X_\sigma$  isomorphic to $S^1$.  To see that $X_\sigma\setminus [s_1] $ is 
			\end{minipage} \hspace{.5cm}
		\begin{minipage}{.35\textwidth}
			\begin{figure}[H]
				\includegraphics[scale=.3]{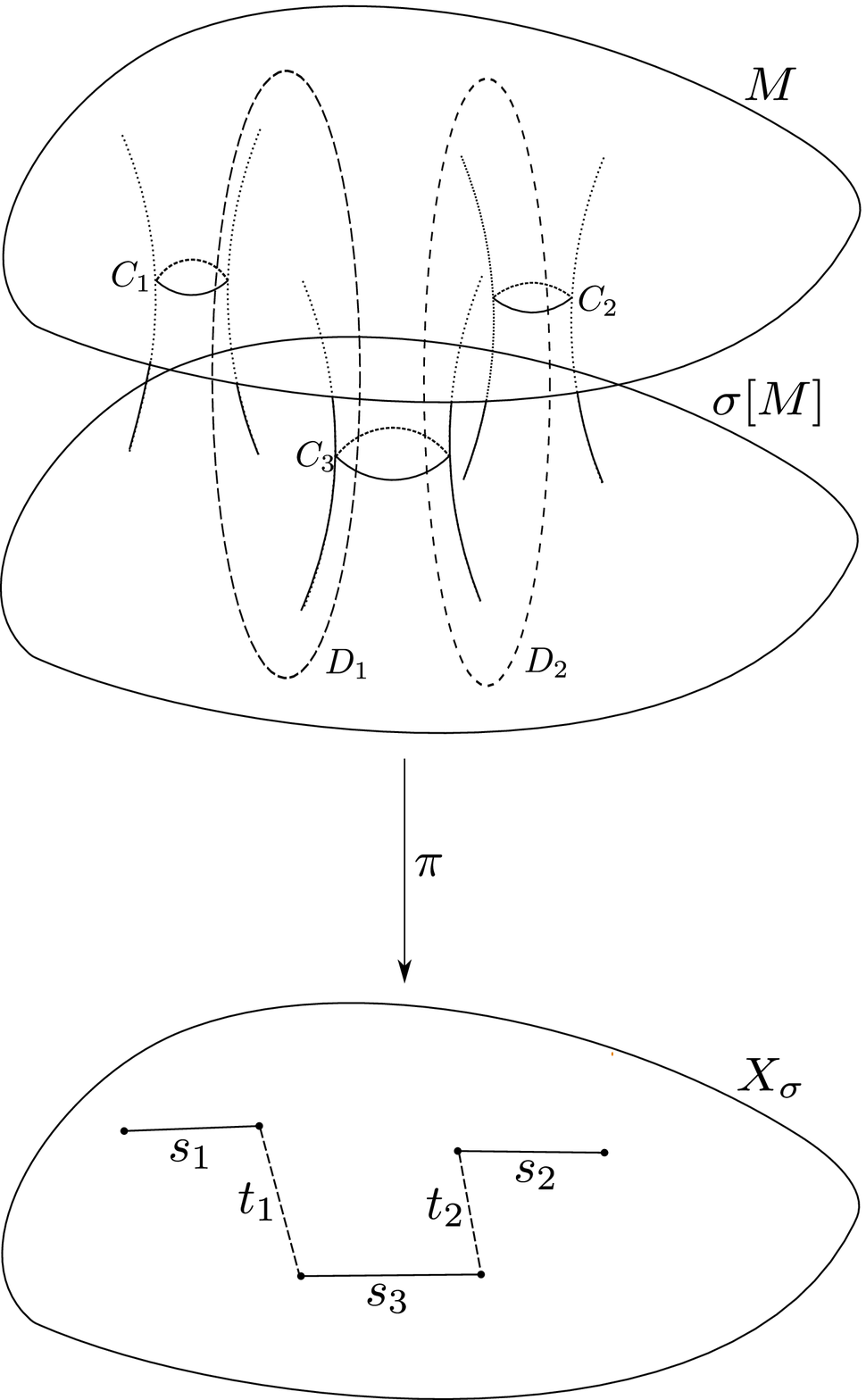} \vspace{-.2cm}
				\caption{}\label{fig:zykel}
			\end{figure} 
		\end{minipage}\\
        path connected, let $\gamma$ be a path in $X_\sigma$ which intersects $\partial N(s_1)$ in the two points $p_1, p_2\in X_\sigma$. Then there is a path $\tilde\gamma$ such that $\tilde\gamma|_{X_\sigma\setminus N(s_1)}=\gamma|_{X_\sigma\setminus N(s_1)}$ and such that the points $p_1$ and $p_2$ are connected via a part of $\partial N(s_1)$. Hence $X_\sigma\setminus [s_1]$ is path connected. Like that one can gradually choose $s_2,\dots s_n$. To find a path $s_j$ not intersecting $s_1, \dots s_{j-1}$, consider $X_\sigma\setminus ([s_1]\cup \dots \cup [s_{j-1}])$ which is path connected and repeat the above procedure until all branchpoints are sorted in pairs. The preimage of $s_j$ under $\pi$ yields two paths in $X$ which both connect the preimage of the connected two  branchpoints. These preimages are ramification points of $\pi$ and we denote them as $b_j^1$ and $b_j^2$. A suitable linear combination of the two paths on $X$ then defines a cycle $C_j$ for $j=1,\dots,n$.
		Since $\pi$ is unbranched on $X\setminus r_\pi$, i.e. a homeomorphism, and since $\pi[r_\pi]=b_\pi\subset [s_1]\cup\dots\cup [s_n]$, $\pi^{-1}[X_\sigma\setminus ([s_1]\cup\dots \cup [s_n])]$ consists of two disjoint connected manifolds whose boundaries both are equal to $\pi^{-1}[s_1]\cup\dots \cup \pi^{-1}[s_n]$ and $\sigma$ interchanges those manifolds. We will call them $M$ and $\sigma[M]$.
		Since the $n$ $C$-cycles are the boundary of $M$ respectively $\sigma[M]$, they are homologous to another, i.e. 
		$
		 C_n=-\sum_{i=1}^{n-1} C_i,
		$
		so this construction yields maximal $n-1$ $C$-cycles which are not homologous to each other. These $n$ cycles we orientate as the boundary of the Riemann surface $M$. We will see later on that, due to the intersection numbers, the cycles $C_1,\dots, C_{n-1}$ are not homologous to each other. 
By construction, each cycle $C_j$ contains the two ramification points $b_j^1$ and $b_j^2$ of $\pi$ and no other ramification points. \\
		The next step is to construct $n-1$ $D$-cycles such that one has 
		$C_j \star D_k = \delta_{jk}$.
    We will see that it is possible to connect $\pi(b_j^2)$ with $\pi(b_{j+1}^1)$ by a path $t_j$ for $j=1,\dots, n-1$ such that $t_j\cap t_k=\emptyset$ for $j\neq k$. 
    Since $X_\sigma\setminus ([s_1]\cup\dots\cup [s_n])$ is path connected, also $X_\sigma\setminus ([s_1]^\circ\cup [s_2]^\circ \cup [s_3]\cup\dots\cup [s_n])$ is path connected. So one can connect $b_1^2$ with $b_2^1$ with a path $t_1$ in $X_\sigma$ not intersecting $s_3,\dots,s_n$ and the path $s_1 + t_1 +s_2$ in $X_\sigma$ contains no loop. As above, one can chose a small open neighborhood $N([s_1]\cup[t_1]\cup [s_2])$ with boundary isomorphic to $S^1$. Therefore, $X_\sigma\setminus ([s_1]\cup\dots\cup [s_n]\cup [t_1])$ is path connected. Repeating this procedure shows that $X_\sigma\setminus ([s_1]\cup \dots \cup [s_n]\cup [t_1]\cup\dots\cup [t_j])$ remains path connected and that $\sum_{m=1}^j (s_m+t_m)+s_{j+1}$ contains no loop for $j=1,\dots,n-1$. This yields the desired $n-1$ paths $t_j$ in $X_\sigma$.
	Lifting these paths via $\pi$ yields each $n-1$ paths on $M$ and $n-1$ paths on $\sigma[M]$. The paths on $M$ and $\sigma[M]$ which result from the lift of $t_j$ both start at $b_j^2$ and end in $b_{j+1}^1$. Hence identifying these end points with each other yields a cycle on $X$ which we denote as $\widetilde D_j$. We orientate $\widetilde D_j$ such that $C_j\star  \widetilde D_j=1$ and $C_{j+1}\star\widetilde D_j=-1$ for $j\in \{1,\dots,n-1\}$. Due to the construction of $\widetilde D_j$ one has $C_{i}\star \widetilde D_j=0$ for $i\not\in\{j,j+1\}$.  Defining $D_j:=\sum_{i=j}^{n-1} \widetilde D_i$ yields for $k<j$
\begin{align*}
C_j\star D_j & = C_j \star \sum\limits_{l=j}^{n-1} \widetilde D_l = C_j\star \widetilde D_j = 1, \quad 
C_k\star D_j = C_k \star \sum\limits_{l=j}^{n-1} \widetilde D_l = 0\\
C_j \star D_k & = C_j \star \sum\limits_{l=k}^{n-1} \widetilde D_l = C_j \star (\widetilde D_j + \widetilde D_{j-1}) = 1-1=0
\end{align*}
and hence $n-1$ cycles which obey $C_k\star D_j=\delta_{kj}$. Two cycles can not be homologeous to each other if the intersection number of each one of those cycles with a third cycle is not equal, hence $C_k\star D_j=\delta_{kj}$ implies that the above construction yields $2n-2$ cycles $C_j$ and $D_j$ which are not homologous to each other.
		To construct the missing $4g_\sigma$ cycles,
		we choose a symplectic cycle basis $A_{\sigma,i}, B_{\sigma,i}$ of $H_1(X_\sigma,\Z)$ such that they intersect none of the paths $s_1,\dots,s_n$ and $t_1,\dots, t_{n-1}$. 
        This is possible since all of these paths in $X_\sigma$ are connected and hence can be contracted to a point. On the preimage of $X_\sigma\setminus\{\bigcup_{j=1}^{n-1}([s_j]\cup [t_j])\cup [s_n]\}$, the map $\pi$ is a homeomorphism. So each of the cycles in $H_1(X_\sigma,\Z)$ is lifted to one cycle in $M$ and one cycle in $\sigma[M]$ via $\pi$ and those two cycles are interchanged by $\sigma$. Thus 
        lifting the whole basis yields  $4g_\sigma$ cycles on $X$ where we denote the $2g_\sigma$ cycles lifted to $M$ as $A_i$ and $B_i$ and the corresponding cycles lifted to $\sigma[M]$ as $\sigma_\sharp A_i$ and $\sigma_\sharp B_i$. Then these cycles obey the desired transformation behaviour under $\sigma_\sharp$. Since $M$ and $\sigma[M]$ are disjoint, the intersection number of the lifted cycles on $X$ stays the same as the intersection number of the corresponding cycles on $X_\sigma$ if two cycles are lifted to the same sheet $M$ respectively $\sigma[M]$ or equals zero if they are lifted to different sheets. Furthermore, the construction of these cycles ensured that the lifted $A$- and $B$-cycles do not intersect any of the $C$- and $D$-cycles on $X$. Hence $A_i,\sigma_\sharp A_i, B_i, \sigma_\sharp B_i, C_j$ and $D_j$ are in total $4g_\sigma+2n-2$ cycles which obey condition \eqref{eqn:canonical_zero_prym}. So by Hurwitz Formula, they represent a symplectic basis of $H_1(X,\Z)$ and the $A$- and $C$-cycles are disjoint. That the $C$- and $D$-cycles constructed like this have the desired transformation behavior under $\sigma_\sharp$ is shown in the next lemma.
	\begin{lem}\label{lem:cd-cycles} For $C_j,D_j\in H_1(X,\Z)$ as defined above one has
						$\sigma_\sharp  C_j=- C_j$ and $\sigma_\sharp  D_j=- D_j$.
		\end{lem}
		\begin{proof}
Every cycle $C_j$ is the preimage of a path in $X_\sigma$ and $X_\sigma$ is invariant under $\sigma$. So $\sigma[C_j]=[C_j]$ and the two points $b_j^1$ and $b_j^2$ stay fixed. Therefore $\sigma_\sharp C_j=\pm C_j$.
Since $\sigma$ commutes the two lifts of the path $s_j$ in $X_\sigma$, i.e. $b_j^1$ and $b_j^2$ are the only fixed points of $\sigma$ on $C_j$, one has $\sigma_\sharp C_j=-C_j$.
By the same means, since $D_j$ also consists of the two lifts of $t_j$ which are interchanged by $\sigma$, one also  has $\sigma_\sharp D_j=-D_j$. 
\end{proof}
\section{Decomposition of $H_1(X,\Z)$}
With help of the Abel map $\mathrm{Ab}$ one can identify the elements of $H_1(X,\Z)$ with a lattice in $\C^g$ such that $\Jac(X)\simeq \C^g/\Lambda$, compare \cite[Section VIII.2]{Mir}. To do so, let $\omega_1,\ldots, \omega_{g}\in H^0(X,\Omega)$ be a basis of the $g=2g_\sigma+n-1$ holomorphic differential forms on $X$ which are normalized with respect to the $A$-, $\sigma_\sharp A$- and $C$-cycles, i.e. 
	 	\begin{equation}\label{eqn:hol_norm}
	 	\oint\limits_{A_i}\omega_\ell = \delta_{i\ell}, \quad \oint\limits_{\sigma_\sharp  A_i}\omega_{g_\sigma+\ell} = \delta_{i\ell}, \quad 
	 	\oint\limits_{C_j}\omega_{2g_\sigma+k} =  \delta_{jk} 
	 	\end{equation}
	 	and all other integrals over one of the $A$- and $C$-cycles with another element of the basis of $H^0(X,\Omega)$ are equal to zero. Furthermore, note that the construction of the $A$-cycles yields $\sigma^\ast\omega_i=\omega_{g_\sigma+i}$ for $i=1,\dots, g_\sigma$ and that Lemma \ref{lem:cd-cycles} implies $\sigma^\ast\omega_{2g_\sigma+j}=-\omega_{2g_\sigma+j}$.
        We define 
        \begin{equation}\label{eqn:omegatransformed}
        \omega_i^\pm:=\frac{1}{2}(\omega_i\pm\omega_{g_\sigma+i}) \quad \text{ and } \quad \omega^-_{g_\sigma+j}:=\omega_{2g_\sigma+j}.
        \end{equation}
        Direct calculation shows that these differential forms also yield a basis of $H^0(X,\Omega)$. For a path $\gamma$ in $X$, we define the vectors
	 		 \begin{align*}
	 		 \Omega_{\gamma} := \bigg(
	 		 \int\limits_{\gamma} \omega_k 
	 		 \bigg)_{k=1}^g, \	 \Omega_{\gamma}^+ := 
	 		 \bigg(
	 	     \int\limits_{\gamma} \omega_k^+ 
	 		 \bigg)_{k=1}^{g_\sigma}, \
	 		 	 \Omega_{\gamma}^- := 
	 		 	 \bigg(
	 		 	 \int\limits_{\gamma} \omega_k^-
	 		 	\bigg)_{k=1}^{g_\sigma+n-1}.
	 		 \end{align*}
and the following lattices generated over $\Z$ as
	 	\begin{equation}\label{eqn:lattice}
	 		 \begin{aligned}
	 		 \Lambda & := \langle \Omega_{ A_i}, \Omega_{\sigma_\sharp  A_i}, \Omega_{ C_j},\Omega_{ B_i},\Omega_{\sigma_\sharp  B_i}, \Omega_{ D_j}\rangle_{\substack{
	 		 	i =1,\dots, g_\sigma\phantom{dk}\\ j=1,\dots, n-1
	 		 	}}\\
	 		 \Lambda_+ & := \langle \Omega_{{A}_i+\sigma_\sharp {A}_i}^+, \Omega_{{B}_i+\sigma_\sharp {B}_i}^+\rangle_{i=1,\dots,g_\sigma}\\ 
	 		 \Lambda_- & := \langle\Omega_{{A}_i-\sigma_\sharp {A}_i}^-, \Omega_{{B}_i-\sigma_\sharp {B}_i}^-, \Omega_{ C_j}^-, \Omega_{{D}_j}^-\rangle_{\substack{
	 		 		i =1,\dots, g_\sigma\phantom{dk}\\ j=1,\dots, n-1
	 		 	}}. \\
	 		 \end{aligned}
	 		 \end{equation}
	Furthermore, the mapping
	\begin{equation*}
	\Phi:\C^{g} \to \C^{g_\sigma} \oplus \C^{g_\sigma+n-1}, \ \begin{pmatrix}
	v_1\\ \vdots \\ v_g
	\end{pmatrix} \mapsto
    \begin{pmatrix}
	\frac{1}{2}(v_1+v_{g_\sigma+1}) \\ \vdots \\ \frac{1}{2}(v_{g_\sigma}+v_{2g_\sigma}) 
	\end{pmatrix}
	\oplus 
	\begin{pmatrix}
	\frac{1}{2}(v_1-v_{g_\sigma+1}) \\ \vdots \\ \frac{1}{2}(v_{g_\sigma}-v_{2g_\sigma}) \\ v_{2g_\sigma+1} \\ \vdots \\ v_{2g_\sigma +n-1}
	\end{pmatrix}
	\end{equation*}
	is obviously linear and bijective. Hence $\Phi$ is a vector space isomorphism.
\begin{lem}\label{lem:iso}
	 For every path $\gamma$ on $X$ one has
	\[\Phi(\Omega_\gamma) = \Omega^+_{\gamma}\oplus \Omega^-_\gamma = \Omega^+_{\frac{1}{2}(\gamma+\sigma_\sharp\gamma)}\oplus \Omega^-_{\frac{1}{2}(\gamma-\sigma_\sharp\gamma)}.\]
\end{lem}
\begin{proof}
	The first equality follows immediately from the definition of $\Phi$ and the differential forms in \eqref{eqn:omegatransformed}:
	\begin{align*}
	\Phi(\Omega_\gamma) = \Phi\begin{pmatrix}
	\int_\gamma \omega_1 \\ \vdots \\ \int_\gamma \omega_g
	\end{pmatrix} 
&	= \begin{pmatrix}
	\frac{1}{2}(\int_\gamma \omega_1+\omega_{g_\sigma+1}) \\ \vdots \\ 
	\frac{1}{2}(\int_\gamma \omega_{g_\sigma}+\omega_{2g_\sigma}) 
		\end{pmatrix} 
	\oplus \begin{pmatrix}
	\frac{1}{2}(\int_\gamma \omega_1-\omega_{g_\sigma+1}) \\ \vdots \\ 
	\frac{1}{2}(\int_\gamma \omega_{g_\sigma}-\omega_{2g_\sigma}) \\
	\int_\gamma \omega_{2g_\sigma+1}\\
	\int_\gamma \omega_{2g_\sigma+n-1} 
	\end{pmatrix}
\\&	= \begin{pmatrix}
	\int_\gamma \omega^+_1 \\ \vdots \\ 
	\int_\gamma \omega^+_{g_\sigma}
	\end{pmatrix} 
	\oplus \begin{pmatrix}
	\int_\gamma \omega^-_1 \\ \vdots \\ 
	\int_\gamma \omega^-_{g_\sigma} \\
	\int_\gamma \omega^-_{g_\sigma+1}\\
	\int_\gamma \omega^-_{g_\sigma+n-1} 
	\end{pmatrix}.
	\end{align*}
Since $\omega_k^+=\sigma^\ast\omega_k^+$ for $k=1,\dots, g_\sigma$ and $\omega_k^-=-\sigma^\ast\omega_k^-$ for $k=1,\dots,g_\sigma+n-1$ one has  
\[
\int_\gamma \omega_k^+ = \frac{1}{2} \bigg(\int_\gamma \omega_k^++\sigma^\ast \omega_k^+\bigg) = \int_{\frac{1}{2}(\gamma+\sigma_\sharp\gamma)}\omega_k^+ 
\]
as well as
\[
\int_\gamma \omega_k^- = \frac{1}{2} \bigg(\int_\gamma \omega_k^--\sigma^\ast \omega_k^-\bigg) = \int_{\frac{1}{2}(\gamma-\sigma_\sharp\gamma)}\omega_k^- 
\]
which implies the second equality.
\end{proof}
\begin{cor}\label{cor:lattice}
	The generators of $\Phi^{-1}(\Lambda_+\oplus \Lambda_-)$ span a basis of $\C^g$ over $\R$, the generators of $\Lambda_+$ span a basis of $\C^{g_\sigma}$ over $\R$ and the generators of $\Lambda_-$ span a basis of $\C^{g_\sigma+n-1}$ over $\R$. 	
\end{cor}
\begin{proof}
		Since $\Jac(X)=\C^g/\Lambda$ is a complex torus, the generators of $\Lambda$ given in \eqref{eqn:lattice} are a basis of $\C^g$ over $\R$, compare for example \cite[Section II.2]{LB}. Basis transformation yields that $\Omega_{A_i+\sigma_\sharp A_i}$, $\Omega_{A_i-\sigma_\sharp A_i}$,$\Omega_{B_i+\sigma_\sharp B_i}$, $\Omega_{B_i-\sigma_\sharp B_i}$ , $\Omega_{C_j}$ and $\Omega_{D_j}$ are also a basis of $\C^g$ over $\R$. Since $\Phi$ is a vector space ismorphism with  
		\begin{align*}
		\Phi(\Omega_{A_i+\sigma_\sharp A_i}) = \Omega^+_{A_i+\sigma_\sharp A_i} \oplus 0, &&
		\Phi(\Omega_{A_i-\sigma_\sharp A_i}) = 0 \oplus  \Omega^-_{A_i-\sigma_\sharp A_i}, \\
		\Phi(\Omega_{B_i+\sigma_\sharp B_i}) = \Omega^+_{B_i+\sigma_\sharp B_i} \oplus 0, &&
		\Phi(\Omega_{B_i-\sigma_\sharp B_i}) = 0\oplus \Omega^-_{B_i-\sigma_\sharp B_i}, \\
		\Phi(\Omega_{C_j})= 0 \oplus \Omega^-_{C_j}, && \Phi(\Omega_{D_j})= 0 \oplus \Omega^-_{D_j},
		\end{align*}
		the generators of $\Phi^{-1}(\Lambda_+\oplus \Lambda_-)$ yield a basis of $\C^g$ over $\R$. Since $\Phi$ is an isomorphism, the generators of $\Lambda_+$ are a basis of $\C^{g_\sigma}$ and the generators of $\Lambda_-$ of $\C^{g_\sigma+n-1}$ over $\R$.
	\end{proof}
In the sequel, we will apply $\Phi$ and $\Phi^{-1}$ to lattices. Note that, for shortage of notation, we abuse the notation in the sense that $\Phi(\Lambda)$ denotes the lattice in $\C^{g_\sigma}\oplus \C^{g_\sigma+n-1}$ spanned by the image of the generators of $\Lambda$ under $\Phi$ and analogously for $\Phi^{-1}$ applied to lattices. 
	 		\begin{lem} \label{lem:lattice_translate}
	 			\begin{enumerate}[(a)]
	 				\item $\Lambda_+\oplus\, 0= \Phi(\Lambda)\cap (\C^{g_\sigma}\oplus 0)$, $0\,\oplus \Lambda_-= \Phi(\Lambda)\cap (0\oplus \C^{g_\sigma+n-1})$.
\item
	 			\begin{equation}\label{eqn:lambda_decomposed}
	 		     \Phi(\Lambda) = (\Lambda_+ \oplus \Lambda_-) + M
	 			\end{equation} with
	 			\begin{multline*} 
	 			M:=\Big\{\sum\limits_{i=1}^{g_\sigma} \Big(\frac{a_i}{2}\Omega^+_{A_i+\sigma_\sharp  A_i}+ \frac{b_i}{2}\Omega^+_{ B_i+\sigma_\sharp  B_i}\Big) \oplus \\ \oplus \sum\limits_{i=1}^{g_\sigma} \Big( \frac{a_i}{2}\Omega^-_{A_i-\sigma_\sharp  A_i} + \frac{b_i}{2}\Omega^-_{ B_i-\sigma_\sharp  B_i}\Big) \, \Big| \, a_i, b_i\in\{0,1\}\Big\}.
	 			\end{multline*}
	 			\item $M \cap (\Lambda_+ \oplus \Lambda_-)=\{0\}$
	 			\end{enumerate}
	 		\end{lem}
	 		\begin{proof}
	 		 Obviously, $\Lambda_+\oplus 0$ is contained in $\Phi(\Lambda)\cap (\C^{g_\sigma}\oplus 0)$. To see that $\Phi(\Lambda)\cap (\C^{g_\sigma}\oplus 0)$ is also a subset of $\Lambda_+\oplus 0$, note that for every $\gamma\in \Lambda$ there exists coefficients $a_i,a_{\sigma,i},b_i,b_{\sigma_i},c_j,d_j\in \Z$ such that 
	 				\[
	 				\gamma = \sum\limits_{i=1}^{g_\sigma} a_i \Omega_{A_i}+a_{\sigma,i}\Omega_{\sigma_\sharp A_i}+ b_i \Omega_{B_i}+b_{\sigma,i}\Omega_{\sigma_\sharp B_i}+c_j \Omega_{C_j}+d_j\Omega_{D_j}.
	 				\]
	 		The generators of $\Lambda_+$ and $\Lambda_-$ are linearly independent, compare Corollary \ref{cor:lattice}.	
			So the second equality in Lemma \ref{lem:iso} shows that $\Phi(\gamma)\in\C^{g_\sigma}\oplus 0$ can only hold if $c_j=d_j=0$, $a_i=a_{\sigma_i}$ and $b_i=b_{\sigma_i}$. Then for such $\gamma$ it is
				\begin{align*}
				\Phi(\gamma) & = 2a_i\Omega^+_{\frac{1}{2}(A_i+\sigma_\sharp A_i)} + 2b_i\Omega^+_{\frac{1}{2}(B_i+\sigma_\sharp B_i)} \oplus 0 \\ & = a_i\Omega^+_{A_i+\sigma_\sharp A_i} + b_i\Omega^+_{B_i+\sigma_\sharp B_i} \oplus 0 \in \Lambda_+ \oplus 0.
				\end{align*}
				The equality $0\oplus\Lambda_-=\Phi(\Lambda)\cap(0\oplus \C^{g_\sigma+n-1})$ follows in the same manner. So the first part holds.  \\
 			    To get insight into the second part, we will show that for the set of cosets one has
	 			\begin{align*}
	 			\Phi(\Lambda)/\Lambda_+\oplus \Lambda_- & = \{\lambda + (\Lambda_+\oplus \Lambda_-)\, | \,\Phi(\lambda)\in \Lambda \} \\ & = \{\lambda + (\Lambda_+\oplus \Lambda_-)\, | \,\Phi(\lambda)\in M \}.
	 			\end{align*}	
				The lattice $\Lambda$ is a finitely generated abelian group, so also $\Phi(\Lambda)$, $\Lambda_+$ and $\Lambda_-$ are finitely generated abelian groups and   
	 			$\Phi(2\Lambda)\subset\Lambda_+\oplus \Lambda_-\subset \Phi(\Lambda)$,
	 			where the second inclusion is obvious and the first inclusion holds since any element $2\gamma$ of $2\Lambda$ be decomposed as 
				$2\gamma=2\bigr(\frac{1}{2}(\gamma+\sigma_\sharp \gamma)+\frac{1}{2}(\gamma-\sigma_\sharp \gamma)\bigr)$. 
	 			Therefore,
	 			$\Phi(\Lambda)/(\Lambda_+\oplus \Lambda_-) \subset \Phi(\Lambda)/\Phi(2\Lambda)$ and the set of the $(\Lambda:2\Lambda)=2^{2g}$ elements contained in $\Phi(\Lambda)/\Phi(2\Lambda)$ is the maximal set of points which are not contained in $\Lambda_+\oplus \Lambda_-$ but in $\Phi(\Lambda)$. 
	 			One has $\Phi(\Omega_{C_j}),\Phi(\Omega_{D_j})\in \Lambda_+\oplus \Lambda_-\subset \Phi(\Lambda)$. 
	 			Therefore, all points in $M$ are linear combinations of $
	 			\Phi(\Omega_{A_i}),\Phi(\Omega_{\sigma_\sharp A_i}), \Phi(\Omega_{B_i})$ and $\Phi(\Omega_{\sigma_\sharp B_i})$ 
	 			 with coefficients in $\{0,1\}$.
	 		Since $\Omega_{A_i}=\Omega_{A_i+\sigma_\sharp A_i} - \Omega_{\sigma_\sharp A_i}$, one has that $[\Phi(\Omega_{A_i})]=[\Phi(\Omega_{\sigma_\sharp A_i})]$ and $[\Phi(\Omega_{B_i})]=[\Phi(\Omega_{\sigma_\sharp B_i})]$ in $\Phi(\Lambda)/(\Lambda_+\oplus \Lambda_-)$ 
	 			and thus
	 			\begin{equation} \label{eqn:subset_n}
	 				M\subseteq \Big\{\sum_{i=1}^{g_\sigma}a_i\Phi(\Omega_{A_i})+ b_i \Phi(\Omega_{B_i})\,\big|\, a_i,b_i\in \{0,1\}\Big\}.
	  			\end{equation}
	 			Furthermore, $\Phi(\Omega_{A_i})=\frac{1}{2}(\Phi(\Omega_{A_i+\sigma_\sharp A_i})+\Phi(\Omega_{A_i-\sigma_\sharp A_i}))$. 
	 			Due to Corollary \ref{cor:lattice}, these representations of $\Phi(\Omega_{A_i})$ as a vector in $\C^g$ in the basis given by the generators of $\Lambda_+\oplus \Lambda_-$ is unique, i.e.  
$\Phi(\Omega_{A_i})\not \in \Lambda_+\oplus \Lambda_-$ and by the same means $\Phi(\Omega_{\sigma_\sharp A_i}), \Phi(\Omega_{B_i}), \Phi(\Omega_{\sigma B_i})\not \in \Lambda_+\oplus \Lambda_-$.
The linear independence of the generators of $\Lambda$ then yields equality in \eqref{eqn:subset_n}.
	Hence $(\Phi(\Lambda):\Lambda_+\oplus \Lambda_-)=2^{2g_\sigma}$ and so $\Lambda$ can be seen as finitely many copies of $\Lambda_+\oplus \Lambda_-$ translated by the points in $M$. 
    The linear independence of the generators of $\Lambda_+$ and $\Lambda_-$ and the definition of $M$ imply (c).
\end{proof}
 	\begin{rem}
 	In \cite{Mu} it was shown that $\Jac(X_\sigma)\simeq \C^{g_\sigma}/\Lambda_+$ and that the Prym variety $P(X,\sigma)$ can be identified with $\C^{g_\sigma+n-1}/\Lambda_-$. Furthermore, it was also shown that the direct sum $\Jac(X_\sigma)\oplus P(X,\sigma)$ is only isogenous to $\Jac(X)$, but that the quotient of this direct sum divided by a finite set of points is isomorphic to $\Jac(X)$.  	
 	The explicit calculations in Lemmata \ref{lem:iso} and \ref{lem:lattice_translate} are mirroring this connection and the finite set of points which are divided out of the direct sum in \cite[Section 2, Data II]{Mu} are exactly the points in $M$. 
 	\end{rem}	 	
\section{The fixed points of $\sigma$ and the linear equivalence}
\begin{thm}
Let $X$ be a Riemann surface of genus $g$, $K$ a canonical divisor on $X$, $\sigma: X \to X$ a holomorpic involution and $P_1,P_2\in X$ fixed points of $\sigma$. 
Then there exists a divisor $D$ of degree $g$ on $X$ which solves
\begin{equation}\label{eqn:divisor_condition_fixed_points}
D+\sigma(D) \simeq  K+P_1 + P_2
\end{equation} 
if and only if $\sigma$ has exactly the two fixed points $P_1$ and $P_2$.
\end{thm}
\begin{proof}
Assume that $\sigma$ has more fixed points then $P_1$ and $P_2$, i.e. $n>1$, and that \eqref{eqn:divisor_condition_fixed_points} holds. Due to Lemma \ref{lem:divisor} there exists a divisor $\widetilde K$ of degree $2g_\sigma-2$ on $X$ such that $K= \widetilde K+\sigma(\widetilde K)+R_\pi$ and hence equation \eqref{eqn:divisor_condition_fixed_points} yields 
$ D-\widetilde K + \sigma(D-\widetilde K) \simeq  R_\pi+P_1+P_2$.
We sort the $2n$ ramification points in $r_\pi$ into pairs as it was done in the construction of the $C$-cycles and denote the two fixed points on $C_n$ as $P_1$ and $P_2$. Then equation \eqref{eqn:divisor_condition_fixed_points} reads as
$
D-\widetilde K +\sigma(D-\widetilde K) \simeq \sum_{j=1}^{n-1}(b_j^1+b_j^2) + 2P_1 + 2P_2
$.
With $\widetilde D:= D-\widetilde K+\sum_{j=1}^{n-1}b_j^1 -P_1-P_2$
this is equivalent to 
\begin{equation}\label{eqn:divisor_decomposition_2}
\widetilde D+\sigma(\widetilde D)+\sum\limits_{j=1}^{n-1} (b_j^1-b_j^2)\simeq 0.
\end{equation}
Furthermore, $\deg(\widetilde D+\sigma(\widetilde D)+\sum_{j=1}^{n-1} (b_j^1-b_j^2))=0$ and $\deg(\sum_{j=1}^{n-1}(b_j^1-b_j^2))=0$. Since $\deg$ acts linear on divisors and is invariant under $\sigma$, this yields $\deg(\widetilde D)=0$. So counted without multiplicity, there are as many points with positive sign as with negative sign in $\widetilde D$, i.e. $\widetilde D=\sum_{k=1}^\ell (p_k^1-p_k^2)$. Let $\gamma_k:[0,1]\to X$ be a path with $\gamma_k(0)=p_k^1$ and $\gamma_k(1)=p_k^2$. Then $\sigma_\sharp \gamma_k:[0,1]\to X$ is a path with $\sigma(\gamma_k(0))=\sigma(p_k^1)$ and $\sigma(\gamma_k(1))=\sigma(p_k^2)$. Then   define $\gamma_{\widetilde D}:=\sum_{k=1}^\ell \gamma_k$ and  $\sigma_\sharp\gamma_{\widetilde D}:=\sum_{k=1}^\ell \sigma_\sharp \gamma_k$.   
Analogously, let $\gamma_{R,j}$ be defined as the paths $\gamma_{R,j}:[0,1]\to X$ such that $\gamma_{R,j}(0)=b_j^1$ and $\gamma_{R,j}(1)=b_j^2$ for $j=1,\dots, n-1$. Then, due to the construction of the $C$-cycles, one has $\gamma_{R,j}-\sigma_\sharp \gamma_{R,j}=C_j$. We define $\gamma_R:= \sum_{j=1}^{n-1}\gamma_{R,j}$. 
 Set $\gamma:=\gamma_{\widetilde D}+\sigma_\sharp\gamma_{\widetilde D}+\gamma_R$ and let $\omega_1,\dots, \omega_g$ be the canonical basis of $H^0(X,\Omega)$ normalized with respect to $A$- and $C$-cycles as in \eqref{eqn:hol_norm}. 
Again we use the identification $\Jac(X)=\C^{{g}}/\Lambda$ via the Abel map $\mathrm{Ab}$ with the basis of holomorphic $1$-forms on $X$ normalized as in \eqref{eqn:hol_norm}. 
Due to \eqref{eqn:divisor_decomposition_2}, the linear equivalence can also be expressed as
\begin{equation*} 
\mathrm{Ab}\Big(\widetilde D+\sigma(\widetilde D)+\sum\limits_{i=1}^{n-1} (b_i^1-b_i^2)\Big)=0\mod \Lambda.
\end{equation*} 
This equation can only hold if $\Omega_\gamma\in \Lambda$. 
Due to Lemma \ref{lem:iso} we can split $\Omega_\gamma\in\C^{g}$ uniquely by considering $\Phi(\Omega_\gamma)=\Omega^+_\gamma \oplus \Omega^-_\gamma$
Due to the decomposition of $\Lambda$ in \eqref{eqn:lambda_decomposed}, $\Omega_\gamma\in \Lambda$ is equivalent to $\Omega^+_\gamma \oplus \Omega^-_\gamma\in (\Lambda_+\oplus \Lambda_-)+M$ as defined in Lemma \ref{lem:lattice_translate}. So we want to show that $\Omega^+_\gamma \oplus \Omega^-_\gamma$ is not contained in any of the translated copies of $\Lambda_+\oplus\Lambda_-$ if $n>1$.
Since it will turn out that it is $\Omega^-_\gamma$ which leads to this assertion, we determine the explicit form of $\Omega^-_\gamma$. For every $\omega^-\in H^0(X,\Omega)$ such that $\sigma^\ast \omega^-=-\omega^-$ one has 
\[
\int\limits_{\gamma_{\widetilde D}+\sigma_\sharp\gamma_{\widetilde D}} 
\omega^- 
= 
\int\limits_{\gamma_{\widetilde D}}\omega^- + \int\limits_{\gamma_{\widetilde D}}  \sigma^\ast\omega^- 
= 
\int\limits_{\gamma_{\widetilde D}}\omega^- - \int\limits_{\gamma_{\widetilde D}} \omega^- =0 
\] 
as well as
\[
2\int\limits_{\gamma_{R,i}} \omega^- = \int\limits_{\gamma_{R,i}} \omega^- - \int\limits_{\gamma_{R,i}} \sigma^\ast\omega^- = 
\int\limits_{\gamma_{R,i}} \omega^- - \int\limits_{\sigma_\sharp \gamma_{R,i}} \omega^-
=\int\limits_{\gamma_{R,i}-\sigma_\sharp \gamma_{R,i}} \omega^- = \oint\limits_{ C_i} \omega^-,
\]
i.e. $\int_{\gamma_{R,i}}\omega^-=\frac{1}{2}\oint_{C_i}\omega^-$. Since $\gamma-\sigma_\sharp\gamma=2\gamma_R$ one has 
\[
\Omega^-_\gamma = \bigg(
\int\limits_{\gamma_R} \omega_i^- 
\bigg)_{k=1
}^{g_\sigma+n-1}=
\frac{1}{2}\bigg(
\sum\limits_{k=1}^{n-1}\oint_{C_k} \omega_{k}^- 
\bigg)_{k=1}^{g_\sigma+n-1}.
\]
Due to the normalization of the holomorphic $1$-forms defined in \eqref{eqn:hol_norm} one has $\Omega^-_\gamma=\frac{1}{2} \sum_{k=1}^{n-1}\Omega_{C_k}^-$.  
If $\Omega^+_\gamma \oplus \Omega^-_\gamma$ would be contained in one of the translated copies of $\Lambda_+\oplus \Lambda_-$, then $\Omega^-_\gamma$ would be contained in the second component of the direct sum in one of the translated copies of $\Lambda_-$ introduced in Lemma \ref{lem:lattice_translate}. This is not possible, since the generators of $\Lambda_+\oplus \Lambda_-$ are linearly independent and only integer linear combinations of $C$-cycles are contained in all translated lattices. 
Therefore, $\Omega_\gamma\not\in\Lambda$ for $n>1$. For $n\leq 1$, there are no $C$-cycles in $H_1(X, \Z)$ and equation \eqref{eqn:divisor_decomposition_2} would read as $ D+\sigma( D)\simeq 0$. So equation \eqref{eqn:divisor_condition_fixed_points} can only hold if $n\leq 1$. Since $P_1$ and $P_2$ are fixed points of $\sigma$ one has $n=1$. \\ Let now $P_1$ and $P_2$ be the only fixed points of $\sigma$. Then Lemma \ref{lem:divisor} yields that there exists a divisor $\widetilde K$ on $X$ with $\deg(\widetilde K)=2g_\sigma-2$ such that $K= \widetilde K+\sigma(\widetilde K)+P_1+P_2$. Define $D:= \widetilde K+P_1+P_2$. The Hurwitz Formula for $n=1$ yields $\deg(D)=2g_\sigma= g$, compare e.g. \cite[Theorem II.4.16]{Mir}, and one has \enlargethispage{2em}
\[
D+\sigma(D)= \widetilde K+\sigma(\widetilde K)+2P_1+2P_2 \simeq K + P_1 + P_2.
\]
\end{proof}
 
\end{document}